\documentclass[11pt,onecolumn]{IEEEtran}
\usepackage[utf8]{inputenc}
\usepackage{mathtools}
\usepackage{amsmath}
\usepackage{amsthm}
\usepackage{amssymb}
\usepackage{geometry}
\usepackage[bookmarks=false,
 breaklinks=false,pdfborder={0 0 1},backref=false,colorlinks=false]
 {hyperref}

\makeatletter

\@ifundefined{date}{}{\date{}}

\usepackage{amsthm}\usepackage{bm}

\newcommand{\R}{\mathbb{R}}
\newcommand{\E}{\mathbb{E}}
\newcommand{\Cov}{\operatorname{Cov}}
\newcommand{\tr}{\operatorname{tr}}

\newcommand{\calK}{\mathcal K}

\makeatother

\theoremstyle{plain}
\newtheorem{thm}{\protect\theoremname}
\theoremstyle{definition}
\newtheorem{defn}{\protect\definitionname}
\theoremstyle{plain}
\newtheorem{conjecture}{\protect\conjecturename}
\theoremstyle{remark}
\newtheorem{rem}{\protect\remarkname}
\theoremstyle{plain}
\newtheorem{prop}{\protect\propositionname}
\newtheorem{lem}{\protect\lemmaname}
\providecommand{\conjecturename}{Conjecture}
\providecommand{\definitionname}{Definition}
\providecommand{\lemmaname}{Lemma}
\providecommand{\propositionname}{Proposition}
\providecommand{\remarkname}{Remark}
\providecommand{\theoremname}{Theorem}

\begin{document}
\title{Exact Common Information and Exact Channel Synthesis for Correlated
Gaussian Sources}
\author{Lei Yu\thanks{L. Yu is with the School of Statistics and Data Science, LPMC, KLMDASR,
and LEBPS, Nankai University, Tianjin 300071, China (e-mail: leiyu@nankai.edu.cn).
This work was supported by the National Key Research and Development
Program of China under grant 2023YFA1009604 and the NSFC under grant
62101286.} }
\maketitle
\begin{abstract}
In this paper, we resolve two conjectures posed by Yu and Tan in
2020 (in two separate papers published in the IEEE Trans. Inf. Theory).
Specifically, we establish that: 1) the exact common information for
a pair of $\rho$-correlated Gaussian sources is given by the conjectured
expression $\frac{1}{2}\log\frac{1+\rho}{1-\rho}+\frac{\rho}{1+\rho}$;
and 2) the admissible region for the shared randomness rate and the
communication rate in exact channel synthesis is exactly the conjectured
one. These results yield two important consequences. First, for any
$\rho>0$, the exact common information of a correlated Gaussian pair
strictly exceeds Wyner's common information. Second, for $\rho>0$,
the exact channel synthesis of such a pair requires strictly higher
rates than the total-variation version. The proof combines an exact
optimal-transport representation of the worst-case Gaussian cross-entropy,
Fathi's Gaussian transport inequality, and a determinant inequality
arising from the covariance structure of the conditional means.
\end{abstract}

\begin{IEEEkeywords}
Gaussian sources, Exact common information, Exact channel synthesis,
Communication complexity, Transport inequality
\end{IEEEkeywords}

\section{Introduction}

\subsection{Common information }

The common information problem asks for the minimum amount of common
randomness needed to generate two correlated random variables at two
separate terminals. In Wyner's formulation \cite{WynerCI}, Wyner
considered approximate generation in relative entropy and obtained
the common information 
\[
C_{{\rm Wyner}}(X;Y)=\inf_{W:\,X-W-Y}I(X,Y;W).
\]

The exact common information problem, introduced by Kumar, Li, and
El Gamal \cite{KLE2014}, is different. Here the synthesized distribution
must equal the target distribution exactly, rather than merely approach
it asymptotically in a weak metric. The common random variable is
first generated and then independently processed at the two terminals.
The exact common information is the minimum asymptotic rate of this
common randomness. 

Let $\pi_{XY}$ be a probability distribution on a standard Borel
product space. At blocklength $n$, an exact synthesis code consists
of a discrete random variable $W$ and conditional distributions $P_{X^{n}|W},P_{Y^{n}|W},$
such that 
\[
P_{X^{n}Y^{n}}=\sum_{w}P_{W}(w)P_{X^{n}|W=w}P_{Y^{n}|W=w}=\pi^{\otimes n}_{XY},
\]
where the conditional independence relation $X^{n}-W-Y^{n}$ holds. 

The exact common information is the asymptotic normalized entropy
of the smallest such common random variable. The variable-length formulation
is equivalent because for a prefix-free code, $H(W)\leq L(W)<H(W)+1,$
and therefore the difference is negligible after normalization by
$n$. Accordingly, 
\begin{equation}
C_{{\rm Exact}}(\pi_{XY})=\lim_{n\to\infty}\frac{1}{n}\inf\left\{ H(W):P_{X^{n}Y^{n}}=\pi^{\otimes n}_{XY},\ X^{n}-W-Y^{n}\right\} .\label{eq:exact-multiletter}
\end{equation}

A natural interpretation of exact common information can be described via the following thought experiment. Consider a random particle (or planet) that, at a certain instant, decomposes into multiple components. These components inherit shared common randomness and then evolve independently, governed jointly by this common randomness as well as their respective individual randomness. After some time has elapsed, given the observed joint distribution of the components, one may wish to estimate the amount of common randomness they possess. Or conversely, given the quantity of common randomness shared among them, we aim to characterize the set of feasible joint distributions these components can attain. This is exactly the common information problem. 

The exact common information is no smaller than Wyner's common information.
However, unlike Wyner's common information, the single-letter characterization
of exact common information is rarely known, except for doubly symmetric
binary sources (DSBSes) \cite{YuTan2020_exact} and a certain class
of sources satisfying $C_{\mathrm{Exact}}=C_{{\rm Wyner}}$ \cite{KLE2014,Vellambi2018,YuTan2020_exact}.
 As for Gaussian sources, the exact common information is still unknown. 

In addition, Yu and Tan \cite{YuTan2018,yu2020corrections,YuTan2020_exact}
introduced the notion of Rényi common information, which is defined
as the minimum common rate when the the relative entropy is replaced
by more general divergences---the family of Rényi divergences. The
family of Rényi common information unifies Wyner's common information
and exact common information since it includes them as two special
cases with the Rényi order equal to $1$ or $\infty$. 

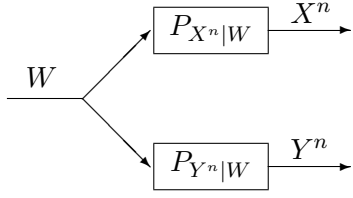
\begin{figure}
\centering \setlength{\unitlength}{0.05cm} 
{ \begin{picture}(100,60) 
\put(5,30){\line(1,0){20}} \put(25,30){\vector(1,1){18}}
\put(25,30){\vector(1,-1){18}} \put(44,42){\framebox(30,12){$P_{X^{n}|W}$}}
\put(44,6){\framebox(30,12){$P_{Y^{n}|W}$}} \put(74,48){\vector(1,0){22}}
\put(74,12){\vector(1,0){22}} \put(10,33){%
\mbox{%
$W$%
}} \put(80,50){%
\mbox{%
$X^{n}$%
}} \put(80,14){%
\mbox{%
$Y^{n}$%
}} \end{picture}}

\label{fig:dss}

\caption{Distributed source synthesis.}
\end{figure}

\subsection{Distributed channel synthesis }

Common information has natural applications in distributed channel
synthesis \cite{Bennett02,Win02,cuff13,bennett14quantum,Harsha10}.
The latter problem, illustrated in Fig. \ref{fig:dcs}, refers to
the problem of determining the minimum communication rate required
to generate a bivariate source $\left\{ \left(X^{n},Y^{n}\right)\right\} _{n\in\mathbb{N}}$
 with $X^{n}$ generated at the encoder and $Y^{n}$ generated at
the decoder such that the induced joint distribution $P_{X^{n}Y^{n}}$
approximately or exactly equals $\pi^{\otimes n}_{XY}$ for all $n\in\mathbb{N}$.
In this paper, we focus on the exact synthesis, equivalently, requiring
$P_{X^{n}Y^{n}}=\pi^{\otimes n}_{XY}$. When there is no shared randomness,
the exact channel synthesis problem reduces to the exact common information
problem. 

Consider the distributed source simulation setup depicted in Fig.
\ref{fig:dcs}. A sender and a receiver share a uniformly distributed
source of randomness\footnote{For simplicity, we assume that $e^{nR_{0}}$ is integers.}
$K\sim\mathrm{Unif}(\calK),\calK:=[e^{nR_{0}}]$. The sender has access
to a memoryless source $X^{n}\sim\pi^{\otimes n}_{X}$ that is independent
of $K_{n}$, and wants to transmit information about the correlation
between correlated sources $\left(X^{n},Y^{n}\right)\sim\pi^{\otimes n}_{XY}$
to the receiver. Given the shared randomness and the correlation
information from the sender, the receiver generates a memoryless source
$Y^{n}\sim\pi^{\otimes n}_{Y|X}(\cdot|X^{n})$. Specifically, given
$X^{n}$ and $K$, the sender generates a ``message'' (i.e., a discrete
random variable) $M$ by a random mapping $P_{M|X^{n}K}$, and then
sends it to the receiver error free. Upon accessing to $K$ and receiving
$M$, the receiver generates a source $Y^{n}$ by a random mapping
$P_{Y^{n}|MK}$. The joint distribution induced by this code is 
\begin{align*}
 & P_{X^{n}KMY^{n}}:=P_{X^{n}}P_{K}P_{M|X^{n}K}P_{Y^{n}|MK}.
\end{align*}
Now we would like to determine the minimum amount of communication
such that $P_{X^{n}Y^{n}}=\pi^{\otimes n}_{XY}$ (or equivalently,
$P_{Y^{n}|X^{n}}=\pi^{\otimes n}_{Y|X}$). 

\begin{defn}
The admissible region of shared randomness rate and communication
rate for the exact channel synthesis problem is defined as 
\begin{align}
 & \mathcal{R}_{\mathrm{Exact}}(\pi_{XY})\nonumber \\
 & =\mathrm{cl}\bigcup_{n\ge1}\left\{ \begin{array}{l}
(R_{0},R):\exists(P_{M|X^{n}K},P_{Y^{n}|MK})\textrm{ s.t.}\\
\qquad P_{Y^{n}|X^{n}}=\pi^{\otimes n}_{Y|X},\\
\qquad R\ge\frac{1}{n}H(M|K)
\end{array}\right\} .\label{eq:-6}
\end{align}
\end{defn}
In contrast to exact channel synthesis, total variation (TV)-approximate
synthesis only requires the TV distance between the empirical distribution
$P_{X^{n}Y^{n}}$ and the target distribution $\pi^{n}_{XY}$ to vanish
asymptotically. Early studies by Bennett et al. \cite{Bennett02}
and Winter \cite{Win02} investigated exact and TV-approximate channel
synthesis under the assumption of unlimited shared encoder--decoder
randomness, proving that the minimal asymptotic communication rate
for both synthesis schemes equals the mutual information $I(X;Y)$
of the target distribution $(X,Y)\sim\pi_{XY}$.

Subsequent research explored the fundamental tradeoff between communication
rate and shared randomness rate in TV-approximate synthesis. Cuff
\cite{cuff13} and Bennett et al. \cite{bennett14quantum} characterized
this rate tradeoff, while Harsha et al. \cite{Harsha10} adopted a
rejection sampling framework to analyze one-shot exact synthesis for
discrete sources, establishing a bounded shared randomness cost with
a mild increment in expected description length. Li and El Gamal \cite{LiElgamal2018}
further refined the finite shared randomness upper bound for discrete
settings with a negligible communication rate penalty. 

Most prior work focused on TV-approximate synthesis or exact synthesis
under extreme randomness conditions, except for \cite{YuTan2020b}.
Yu and Tan \cite{YuTan2020b}  fully characterized the optimal tradeoff
between shared randomness rate and communication rate for exact synthesis
of the doubly symmetric binary source (DSBS), and verified that exact
synthesis requires a strictly higher communication rate than its TV-approximate
counterpart. This is the first example for which the optimal rate
tradeoff is explicitly known. 

\begin{figure*}
\centering \setlength{\unitlength}{0.06cm} { \begin{picture}(140,35)
\put(-5,10){\vector(1,0){30}} \put(-10,13){%
\mbox{%
$X^{n}\sim\pi^{\otimes n}_{X}$%
}} \put(25,4){\framebox(30,12){$P_{M|X^{n}K}$}} \put(55,10){\vector(1,0){30}}
\put(65,13){%
\mbox{%
$M$%
}} \put(85,4){\framebox(30,12){$P_{Y^{n}|MK}$}} \put(115,10){\vector(1,0){20}}
\put(120,13){%
\mbox{%
$Y^{n}\sim\pi^{\otimes n}_{Y|X}(\cdot|X^{n})$%
}} \put(40,27){\vector(0,-1){11}} \put(-5,30){%
\mbox{%
$K\sim\mathrm{Unif}[e^{nR_{0}}]$%
}} \put(100,27){\vector(0,-1){11}} \put(-5,27){\line(1,0){105}}
\end{picture}}

\caption{Exact channel synthesis. }\label{fig:dcs}
\end{figure*}
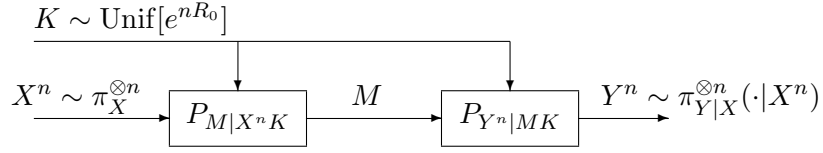

\subsection{The Gaussian setting }

Consider the standard bivariate Gaussian distribution 
\[
\pi_{\rho}=\mathcal{N}(\mathbf{0},\Sigma_{\rho}),
\]
where $0\leq\rho<1$ and 
\[
\Sigma_{\rho}=\begin{pmatrix}1 & \rho\\
\rho & 1
\end{pmatrix}.
\]

For this distribution, Wyner's common information is 
\[
C_{{\rm Wyner}}(\pi_{\rho})=\frac{1}{2}\ln\frac{1+\rho}{1-\rho}.
\]
Indeed, it is achieved by the Gaussian decomposition 
\[
W\sim\mathcal{N}(0,\rho),\qquad X=W+N_{X},\qquad Y=W+N_{Y},
\]
where 
\[
N_{X},N_{Y}\sim\mathcal{N}(0,1-\rho)
\]
are independent of each other and of $W$.

As already mentioned, the exact common information is always at least
Wyner's common information. Yu and Tan \cite{YuTan2020_exact} proved
the upper bound 
\begin{equation}
C_{{\rm Exact}}(\pi_{\rho})\leq\frac{1}{2}\ln\frac{1+\rho}{1-\rho}+\frac{\rho}{1+\rho}.\label{eq:UB}
\end{equation}
They also proved that the exact common information equals the $\infty$-Rényi
common information for this Gaussian source. They conjectured this
upper bound is tight. 
\begin{conjecture}[\cite{YuTan2020_exact}]
\label{conj:ECI} For every $0\leq\rho<1$, 
\[
C_{{\rm Exact}}(\pi_{\rho})=\frac{1}{2}\ln\frac{1+\rho}{1-\rho}+\frac{\rho}{1+\rho}.
\]
\end{conjecture}

The importance of this question is that the second term $\frac{\rho}{1+\rho}$
is strictly positive whenever $\rho>0$. Hence the conjecture would
show that exact synthesis requires strictly more common randomness
than Wyner's approximate synthesis in the Gaussian setting. 

As for exact channel synthesis of Gaussian distributions, Yu and
Tan \cite{YuTan2020b} proved the inner bound 
\begin{equation}
\mathcal{R}_{\mathrm{Exact}}(\pi_{\rho})\supseteq\mathcal{R}(\pi_{\rho}),\label{eq:Gaussian-1}
\end{equation}
where 
\begin{align*}
 & \mathcal{R}(\pi_{\rho}):=\left\{ \begin{array}{rcl}
\left(R,R_{0}\right) & : & \alpha\in[\rho^{2},1],\alpha\beta=\rho^{2},\\
R & \ge & \frac{1}{2}\log\left[\frac{1}{1-\alpha}\right],\\
R+R_{0} & \ge & \frac{1}{2}\log\left[\frac{1-\rho^{2}}{\left(1-\alpha\right)\left(1-\beta\right)}\right]+\frac{\rho\sqrt{\left(1-\alpha\right)\left(1-\beta\right)}}{1-\rho^{2}}
\end{array}\right\} .
\end{align*}
They conjectured this inner bound is tight. 
\begin{conjecture}[\cite{YuTan2020b} ]
\label{conj:ECS} For every $0\leq\rho<1$, 
\[
\mathcal{R}_{\mathrm{Exact}}(\pi_{\rho})=\mathcal{R}(\pi_{\rho}).
\]
\end{conjecture}

\subsection{Main results}

Our first main contribution is the following lower bound on $C_{{\rm Exact}}(\pi_{\rho})$.
Recall $\pi_{\rho}=\mathcal{N}(\mathbf{0},\Sigma_{\rho})$. 
\begin{thm}
\label{thm:GECI} For every $0\leq\rho<1$, 
\[
C_{{\rm Exact}}(\pi_{\rho})\geq\frac{1}{2}\log\frac{1+\rho}{1-\rho}+\frac{\rho}{1+\rho}.
\]
\end{thm}
Combining this lower bound with Yu--Tan's upper bound in \eqref{eq:UB}
confirm Conjecture \ref{conj:ECI} positively. 
\begin{thm}
For every $0\leq\rho<1$, 
\[
C_{{\rm Exact}}(\pi_{\rho})=\frac{1}{2}\ln\frac{1+\rho}{1-\rho}+\frac{\rho}{1+\rho}.
\]
\end{thm}

\begin{rem}
If $\rho<0$, replace $Y$ by $-Y$. Since exact common information
is invariant under deterministic bijections applied separately to
either terminal, $C_{{\rm Exact}}(\pi_{\rho})=C_{{\rm Exact}}(\pi_{|\rho|}).$
Hence the general scalar Gaussian formula is 
\[
C_{{\rm Exact}}(\pi_{\rho})=\frac{1}{2}\log\frac{1+|\rho|}{1-|\rho|}+\frac{|\rho|}{1+|\rho|},\qquad|\rho|<1.
\]
\end{rem}
This theorem states that 
\[
\text{exact CI}=\text{Wyner CI}+\text{exactness penalty},
\]
with 
\[
\text{exactness penalty}=C_{{\rm Exact}}(\pi_{\rho})-C_{{\rm Wyner}}(\pi_{\rho})=\frac{\rho}{1+\rho}.
\]
Thus the exact common information is strictly larger than Wyner's
common information for every $\rho>0$. The gap satisfies 
\[
0<\frac{\rho}{1+\rho}<\frac{1}{2},
\]
and 
\[
\lim_{\rho\uparrow1}\frac{\rho}{1+\rho}=\frac{1}{2}.
\]
In bits per source symbol, the gap is 
\[
\frac{\rho}{1+\rho}\log_{2}e\leq\frac{1}{2}\log_{2}e\approx0.7213.
\]

The proof idea of Theorem \ref{thm:GECI} can be applied to exact
channel synthesis of Gaussian distributions, leading to our second
main contribution. 
\begin{thm}
\label{thm:GECS} For every $0\leq\rho<1$, 
\[
\mathcal{R}_{\mathrm{Exact}}(\pi_{\rho})\subseteq\mathcal{R}(\pi_{\rho}).
\]
\end{thm}
Combining this outer bound with Yu--Tan's inner bound in \eqref{eq:Gaussian-1}
confirm Conjecture \ref{conj:ECS} positively. 
\begin{thm}
For every $0\leq\rho<1$, 
\begin{equation}
\mathcal{R}_{\mathrm{Exact}}(\pi_{\rho})=\mathcal{R}(\pi_{\rho}).\label{eq:Gaussian}
\end{equation}
\end{thm}

\section{Proof of Theorem \ref{thm:GECI}}

In this section, we prove Theorem \ref{thm:GECI} by using the following
strategy: 
\begin{align*}
\text{exact synthesis} & \Longrightarrow\text{multiletter bound}\Longrightarrow\text{optimal transport}\\
 & \quad\Longrightarrow\text{Gaussian }T_{2}\Longrightarrow\text{covariance determinant extremality}.
\end{align*}
The case $\rho=0$ is trivial, and thus, we only consider $\rho\in(0,1)$. 

\subsection{A Multi-letter Bound}

Define 
\[
\pi^{\otimes n}_{\rho}=\mathcal{N}(\mathbf{0},\Sigma^{\otimes n}_{\rho}),
\]
where 
\[
\Sigma^{\otimes n}_{\rho}=\begin{pmatrix}I_{n} & \rho I_{n}\\
\rho I_{n} & I_{n}
\end{pmatrix}.
\]
Equivalently, $X^{n}\sim\mathcal{N}(0,I_{n}),\,Y^{n}\sim\mathcal{N}(0,I_{n}),$and
$\E[X^{n}(Y^{n})^{\mathsf{T}}]=\rho I_{n}.$ We also use $\pi^{\otimes n}_{\rho}$
to denote the corresponding density. 

For two probability measures $\mu,\nu$ on $\R^{n}$, let $\mathcal{C}(\mu,\nu)$
denote the set of all couplings of $\mu$ and $\nu$. Define the maximal
Gaussian cross-entropy 
\begin{equation}
\mathcal{H}^{(n)}_{\rho}(\mu,\nu):=\sup_{Q\in\mathcal{C}(\mu,\nu)}\int Q(dx,dy)\log\frac{1}{\pi^{\otimes n}_{\rho}(x^n,y^n)}.\label{eq:max-cross}
\end{equation}

We now introduce the $n$-letter functional: 
\begin{equation}
\Gamma_{n}(\rho):=\inf\left\{ -h(X^{n}|W)-h(Y^{n}|W)+\E_{W}\mathcal{H}^{(n)}_{\rho}(P_{X^{n}|W},P_{Y^{n}|W})\right\} ,\label{eq:Gamma-def}
\end{equation}
where the infimum is over all discrete $W$ and conditional distributions
satisfying $P_{X^{n}Y^{n}}=\pi^{\otimes n}_{\rho},\,X^{n}-W-Y^{n}.$

Because of conditional independence, 
\[
h(X^{n},Y^{n}|W)=h(X^{n}|W)+h(Y^{n}|W).
\]

We first prove the fundamental converse.
\begin{prop}
\label{prop:multiletterbound}For every exact $n$-letter synthesis
code $(W,P_{X^{n}|W},P_{Y^{n}|W})$ for $\pi^{\otimes n}_{\rho}$,
we have $H(W)\geq\Gamma_{n}(\rho).$ As a consequence, 
\[
C_{{\rm Exact}}(\pi_{\rho})\geq\limsup_{n\to\infty}\frac{1}{n}\Gamma_{n}(\rho).
\]
\end{prop}
\begin{IEEEproof}
Fix $w$ with $P_{W}(w)>0$. Since 
\[
\pi^{\otimes n}_{\rho}(x^n,y^n)=\sum_{w'}P_{W}(w')p_{X^{n}|w'}(x^n)p_{Y^{n}|w'}(y^n)
\]
with $p_{X^{n}|w'},p_{Y^{n}|w'}$ denote the densities of $P_{X^{n}|W=w'},P_{Y^{n}|W=w'}$,
we have pointwise 
\begin{equation}
\pi^{\otimes n}_{\rho}(x^n,y^n)\geq P_{W}(w)p_{X^{n}|w}(x^n)p_{Y^{n}|w}(y^n).\label{eq:pointwise}
\end{equation}
Taking logarithms, 
\[
\log\frac{1}{\pi^{\otimes n}_{\rho}(x^n,y^n)}\leq-\log P_{W}(w)-\log p_{X^{n}|w}(x^n)-\log p_{Y^{n}|w}(y^n).
\]

Let $Q_{w}\in\mathcal{C}(P_{X^{n}|w},P_{Y^{n}|w})$ be arbitrary.
Integrating with respect to $Q_{w}$ gives 
\[
\begin{aligned}\int Q_{w}(dx,dy)\log\frac{1}{\pi^{\otimes n}_{\rho}(x^n,y^n)} & \leq-\log P_{W}(w)\\
 & \quad+h(X^{n}|W=w)+h(Y^{n}|W=w).
\end{aligned}
\]
Since the inequality holds for every coupling $Q_{w}$, 
\[
\mathcal{H}^{(n)}_{\rho}(P_{X^{n}|w},P_{Y^{n}|w})\leq-\log P_{W}(w)+h(X^{n}|w)+h(Y^{n}|w).
\]
Therefore, 
\[
\begin{aligned} & -h(X^{n}|w)-h(Y^{n}|w)+\mathcal{H}^{(n)}_{\rho}(P_{X^{n}|w},P_{Y^{n}|w})\leq-\log P_{W}(w).\end{aligned}
\]
Averaging over $W$, 
\[
-h(X^{n}|W)-h(Y^{n}|W)+\E\mathcal{H}^{(n)}_{\rho}(P_{X^{n}|W},P_{Y^{n}|W})\leq H(W).
\]
Taking the infimum over all exact decompositions proves $H(W)\geq\Gamma_{n}(\rho).$
\end{IEEEproof}

\subsection{Evaluation of the Maximal Gaussian Cross-Entropy}

We now evaluate the functional in \eqref{eq:Gamma-def}. Fix $w$,
and write 
\[
a_{w}=\E[X^{n}|W=w],\qquad b_{w}=\E[Y^{n}|W=w],
\]
and 
\[
A_{w}=\Cov(X^{n}|W=w),\qquad B_{w}=\Cov(Y^{n}|W=w).
\]
Let 
\[
\mu^{0}_{w}=\mathcal{L}(X^{n}-a_{w}|W=w),\qquad\nu^{0}_{w}=\mathcal{L}(Y^{n}-b_{w}|W=w).
\]
The Gaussian density is 
\[
\pi^{\otimes n}_{\rho}(x^n,y^n)=\frac{1}{(2\pi)^{n}(1-\rho^{2})^{n/2}}\exp\left\{ -\frac{|x^n|^{2}+|y^n|^{2}-2\rho\langle x^n,y^n\rangle}{2(1-\rho^{2})}\right\} ,
\]
where $|\cdot|$ is the Euclidian norm. Hence,
\begin{equation}
-\log\pi^{\otimes n}_{\rho}(x^n,y^n)=n\log(2\pi\sqrt{1-\rho^{2}})+\frac{|x^n|^{2}+|y^n|^{2}-2\rho\langle x^n,y^n\rangle}{2(1-\rho^{2})}.\label{eq:gaussian-cross}
\end{equation}

Since the marginals of the coupling are fixed, maximizing the cross-entropy
is equivalent, for $\rho\geq0$, to minimizing $\E\langle X^{n},Y^{n}\rangle.$ 
\begin{lem}[Optimal transport representation]
 For every $w$, 
\begin{equation}
\begin{aligned}\inf_{Q\in\mathcal{C}(P_{X^{n}|w},P_{Y^{n}|w})}\E_{Q}\langle X^{n},Y^{n}\rangle & =a_{w}\cdot b_{w}+\frac{1}{2}\left[W^{2}_{2}(\mu^{0}_{w},-\nu^{0}_{w})-\tr A_{w}-\tr B_{w}\right].\end{aligned}
\label{eq:transport-cross}
\end{equation}
\end{lem}
\begin{IEEEproof}
Under any coupling, 
\[
\begin{aligned}\E\langle X^{n},Y^{n}\rangle & =a_{w}\cdot b_{w}+\E\langle X^{n}-a_{w},Y^{n}-b_{w}\rangle\\
 & =a_{w}\cdot b_{w}+\frac{1}{2}\E|X^{n}-a_{w}+Y^{n}-b_{w}|^{2}-\frac{1}{2}\tr A_{w}-\frac{1}{2}\tr B_{w}.
\end{aligned}
\]
Minimizing $\E\langle X^{n},Y^{n}\rangle$ is therefore equivalent
to minimizing $\E|X^{n}-a_{w}+Y^{n}-b_{w}|^{2}.$ But $X^{n}-a_{w}$
has law $\mu^{0}_{w}$, while $-(Y^{n}-b_{w})$ has law $-\nu^{0}_{w}$.
Hence the minimum of $\E|X^{n}-a_{w}+Y^{n}-b_{w}|^{2}$ is $W^{2}_{2}(\mu^{0}_{w},-\nu^{0}_{w}).$
\end{IEEEproof}
Combining the lemma with \eqref{eq:gaussian-cross} gives the following
exact expression.
\begin{prop}
For every $w$, 
\begin{align}
 & \mathcal{H}^{(n)}_{\rho}(P_{X^{n}|w},P_{Y^{n}|w})=n\log(2\pi\sqrt{1-\rho^{2}})\nonumber \\
 & \qquad+\frac{|a_{w}|^{2}+|b_{w}|^{2}+(1+\rho)(\tr A_{w}+\tr B_{w})-2\rho a_{w}\cdot b_{w}-\rho W^{2}_{2}(\mu^{0}_{w},-\nu^{0}_{w})}{2(1-\rho^{2})}.\label{eq:H-exact}
\end{align}
\end{prop}

\subsection{Fathi's Gaussian Transport Inequality}

We next need to control the Wasserstein term in \eqref{eq:H-exact}.
We use Fathi's Gaussian $T_{2}$ inequality. 
\begin{lem}[Fathi's Gaussian transport inequality \cite{Fathi2018SymmetrizedTalagrand}]
 Let $\mu,\nu$ be probability measures on $\R^{n}$ with finite
second moments, and let $\gamma=\mathcal{N}(\mathbf{0},I_{n}).$ If $\mu$
is centered, then
\[
W^{2}_{2}(\mu,\nu)\le2D(\mu\|\gamma)+2D(\nu\|\gamma).
\]
\end{lem}
More generally, applying the scaling transformation $X\mapsto\sqrt{t}X$
to the underlying random variables for all $\mu,\nu,\gamma$ (and
hence scaling their distributions accordingly), we obtain
\begin{equation}
W^{2}_{2}(\mu,\nu)\le2t\big[D(\mu\|\gamma_{t})+D(\nu\|\gamma_{t})\big],\label{eq:}
\end{equation}
where $\gamma_{t}=\mathcal{N}(0,tI_{n}).$

This scaling is exactly what we need. We fix temporarily $t\in(0,\frac{1-\rho^{2}}{\rho})$
(and lastly, will choose $t=1-\rho$). Since $\mu^{0}_{w}$ and $-\nu^{0}_{w}$
are centered, applying \eqref{eq:} to them yields
\begin{equation}
W^{2}_{2}(\mu^{0}_{w},-\nu^{0}_{w})\leq S_{w}-2tH_{w}+2kt\log(2\pi t),\label{eq:W2-upper}
\end{equation}
where 
\[
S_{w}=\tr A_{w}+\tr B_{w}
\]
and 
\[
H_{w}=h(X^{n}|w)+h(Y^{n}|w).
\]

Substituting \eqref{eq:W2-upper} into \eqref{eq:H-exact} yields
\begin{align}
\mathcal{H}^{(n)}_{\rho}(P_{X^{n}|w},P_{Y^{n}|w})\geq & n\log(2\pi\sqrt{1-\rho^{2}})+\frac{|a_{w}|^{2}+|b_{w}|^{2}+S_{w}-2\rho a_{w}\cdot b_{w}}{2(1-\rho^{2})}\nonumber \\
 & +\frac{\rho t}{1-\rho^{2}}H_{w}-\frac{\rho kt}{1-\rho^{2}}\log(2\pi t).\label{eq:H-lower}
\end{align}

Consequently, the contribution of $w$ to $\Gamma_{n}(\rho)$, denoted
by $g$, satisfies 
\begin{align}
g(w)\geq{} & n\log(2\pi\sqrt{1-\rho^{2}})+\frac{|a_{w}|^{2}+|b_{w}|^{2}+S_{w}-2\rho a_{w}\cdot b_{w}}{2(1-\rho^{2})}\nonumber \\
 & +\left(-1+\frac{\rho t}{1-\rho^{2}}\right)H_{w}-\frac{\rho kt}{1-\rho^{2}}\log(2\pi t).\label{eq:g-lower}
\end{align}

Since $(X^{n},Y^{n})\sim\pi^{\otimes n}_{\rho},$ we have $\E|X^{n}|^{2}=\E|Y^{n}|^{2}=n,$
and $\E[X^{n}\cdot Y^{n}]=n\rho.$ By the law of total covariance,
\[
I_{n}=\E[A_{W}]+\Cov(a_{W}),
\]
and therefore 
\[
\E\left[\tr A_{W}+|a_{W}|^{2}\right]=n.
\]
Similarly, 
\[
\E\left[\tr B_{W}+|b_{W}|^{2}\right]=n.
\]
Finally, conditional independence gives $\Cov(X^{n},Y^{n})=\Cov(a_{W},b_{W}),$
so $\E[a_{W}\cdot b_{W}]=n\rho.$

Hence,
\begin{align*}
 & \E\left[|a_{W}|^{2}+|b_{W}|^{2}+S_{W}-2\rho a_{W}\cdot b_{W}\right]=2k(1-\rho^{2}).
\end{align*}

Averaging \eqref{eq:g-lower} therefore gives 
\begin{align}
\Gamma_{n}(\rho)\geq & n\log(2\pi\sqrt{1-\rho^{2}})+n+\left(-1+\frac{\rho t}{1-\rho^{2}}\right)\E H_{W}-\frac{\rho kt}{1-\rho^{2}}\log(2\pi t).\label{eq:Gamma-t}
\end{align}

At this point, all source-specific information has been reduced to
an upper bound on the conditional entropy $\E H_{W}.$

\subsection{Entropy Bound from the Conditional Covariance }

Define 
\[
U=\Cov(\E[X^{n}|W])=\Cov(a_{W}),
\]
and 
\[
V=\Cov(\E[Y^{n}|W])=\Cov(b_{W}).
\]
Since $\Cov(X^{n},Y^{n})=\rho I_{n}$ and conditional independence
gives $\Cov(X^{n},Y^{n})=\Cov(a_{W},b_{W}),$ we obtain $\Cov(a_{W},b_{W})=\rho I_{n}.$
Thus,
\begin{equation}
\begin{pmatrix}U & \rho I_{n}\\
\rho I_{n} & V
\end{pmatrix}\succeq0.\label{eq:block-positive}
\end{equation}

Moreover, 
\[
I_{n}-U=\E\Cov(X^{n}|W)\succeq0,
\]
and similarly 
\[
I_{n}-V\succeq0.
\]

We now prove the determinant inequality needed for the entropy estimate.
\begin{lem}
For $0<\rho<1$, 
\[
\det(I_{n}-U)\det(I_{n}-V)\leq(1-\rho)^{2k}.
\]
\end{lem}
\begin{IEEEproof}
For $\rho>0$, the positive semidefiniteness in \eqref{eq:block-positive}
implies that $U$ is positive definite. Taking the Schur complement
of $U$, $V-\rho^{2}U^{-1}\succeq0.$ Hence $V\succeq\rho^{2}U^{-1}.$
Since $V\preceq I_{n}$, we have $\rho^{2}U^{-1}\preceq I_{n}.$ Thus
every eigenvalue $u_{i}$ of $U$ satisfies
\begin{equation}
\rho^{2}\leq u_{i}\leq1.\label{eq:-3}
\end{equation}

From $V\succeq\rho^{2}U^{-1}$, we obtain $I_{n}-V\preceq I_{n}-\rho^{2}U^{-1}.$
Therefore,
\[
\det(I_{n}-V)\leq\det(I_{n}-\rho^{2}U^{-1}).
\]

Let $u_{1},\dots,u_{n}$ be the eigenvalues of $U$. Then,
\begin{equation}
\det(I_{n}-U)\det(I_{n}-V)\le\det(I_{n}-U)\det(I_{n}-\rho^{2}U^{-1})=\prod^{n}_{i=1}(1-u_{i})\left(1-\frac{\rho^{2}}{u_{i}}\right).\label{eq:-2}
\end{equation}

For every $u\in[\rho^{2},1]$, 
\[
(1-\rho)^{2}-(1-u)\left(1-\frac{\rho^{2}}{u}\right)=\frac{(u-\rho)^{2}}{u}\geq0.
\]
Consequently, 
\[
(1-u_{i})\left(1-\frac{\rho^{2}}{u_{i}}\right)\leq(1-\rho)^{2}
\]
for every $i$, which gives 
\[
\det(I_{n}-U)\det(I_{n}-V)\leq(1-\rho)^{2k}.
\]
\end{IEEEproof}
\begin{lem}[Conditional entropy bound]
 For every exact Gaussian code, 
\[
\E H_{W}\leq n\log\bigl(2\pi e(1-\rho)\bigr).
\]
\end{lem}
\begin{IEEEproof}
For each $w$, Gaussian maximal entropy gives 
\[
h(X^{n}|w)\leq\frac{n}{2}\log(2\pi e)+\frac{1}{2}\log\det A_{w}.
\]

Averaging and using concavity of $\log\det$, 
\[
\begin{aligned}h(X^{n}|W) & \leq\frac{n}{2}\log(2\pi e)+\frac{1}{2}\log\det\E[A_{W}]\\
 & =\frac{n}{2}\log(2\pi e)+\frac{1}{2}\log\det(I_{n}-U).
\end{aligned}
\]

Similarly, 
\[
h(Y^{n}|W)\leq\frac{n}{2}\log(2\pi e)+\frac{1}{2}\log\det(I_{n}-V).
\]

Therefore 
\begin{equation}
\E H_{W}\leq n\log(2\pi e)+\frac{1}{2}\log\left[\det(I_{n}-U)\det(I_{n}-V)\right].\label{eq:-1}
\end{equation}

Applying the determinant lemma, 
\[
\begin{aligned}\E H_{W} & \leq n\log(2\pi e)+n\log(1-\rho)=n\log\bigl(2\pi e(1-\rho)\bigr).\end{aligned}
\]
\end{IEEEproof}

\subsection{Optimization over the Transport Parameter}

Substituting the conditional entropy bound into \eqref{eq:Gamma-t},
we obtain 
\begin{align}
\Gamma_{n}(\rho)\geq{} & n\log(2\pi e\sqrt{1-\rho^{2}})\nonumber \\
 & +\left(-1+\frac{\rho t}{1-\rho^{2}}\right)n\log\bigl(2\pi e(1-\rho)\bigr)-\frac{\rho kt}{1-\rho^{2}}\log(2\pi t).\label{eq:Gamma-t-final}
\end{align}

For $0<\rho<1$, the derivative of the right-hand side with respect
to $t$ is 
\[
-\frac{\rho n}{1-\rho^{2}}\log\frac{t}{1-\rho}.
\]
Hence the unique maximizer is $t=1-\rho.$ Substituting $t=1-\rho$
gives 
\begin{equation}
\Gamma_{n}(\rho)\geq n\left[\frac{1}{2}\log\frac{1+\rho}{1-\rho}+\frac{\rho}{1+\rho}\right].\label{eq:Gamma-final}
\end{equation}

By the multi-letter converse and \eqref{eq:Gamma-final}, 
\[
C_{{\rm Exact}}(\pi_{\rho})\geq\limsup_{n\to\infty}\frac{\Gamma_{n}(\rho)}{n}\geq\frac{1}{2}\log\frac{1+\rho}{1-\rho}+\frac{\rho}{1+\rho}.
\]

\section{Proof of Theorem \ref{thm:GECS}}

The case $\rho=0$ is trivial, and thus, we only consider $\rho\in(0,1)$. 

\subsection{A Multi-letter Bound}

We now introduce the $n$-letter rate region: 
\begin{align}
\mathcal{R}_{n}(\pi_{\rho}) & :=\left\{ \begin{array}{rcl}
(R_{0},R) & : & \exists P_{W}P_{X^{n}|W}P_{Y^{n}|W}\textrm{ s.t. }\\
P_{X^{n}Y^{n}} & = & \pi^{\otimes n}_{\rho},\\
R & \ge & \frac{1}{n}I(W;X^{n}),\\
R+R_{0} & \ge & \frac{1}{n}\Gamma_{n}(P_{W},P_{X^{n}|W},P_{Y^{n}|W})
\end{array}\right\} ,\label{eq:UB-1}
\end{align}
where 
\begin{equation}
\Gamma_{n}(P_{W},P_{X^{n}|W},P_{Y^{n}|W}):=-h(X^{n}|W)-h(Y^{n}|W)+\E_{W}\mathcal{H}^{(n)}_{\rho}(P_{X^{n}|W},P_{Y^{n}|W}),\label{eq:Gamma-def-1}
\end{equation}
which is the objective function in \eqref{eq:Gamma-def}. We first
prove the fundamental converse.
\begin{prop}
For every exact $n$-letter channel synthesis code $(R_{0},P_{M|X^{n}K},P_{Y^{n}|MK})$
for $\pi^{\otimes n}_{\rho}$ and letting $W=(M,K)$, we have \textup{$X^{n}\text{ — }W\text{ — }Y^{n}$,}
$R\ge\frac{1}{n}I(W;X^{n})$, and $R_{0}+R\ge\frac{1}{n}\Gamma_{n}(P_{W},P_{X^{n}|W},P_{Y^{n}|W})$.
As a consequence, 
\[
\mathcal{R}_{\mathrm{Exact}}(\pi_{\rho})\subseteq\mathrm{cl}\bigcup_{n\ge1}\mathcal{R}_{n}(\pi_{\rho}).
\]
\end{prop}
\begin{IEEEproof}
By examining the proof of Proposition \ref{prop:multiletterbound},
\[
kR+kR_{0}\ge H(W)\ge-h(X^{n}|W)-h(Y^{n}|W)+\E\mathcal{H}^{(n)}_{\rho}(P_{X^{n}|W},P_{Y^{n}|W}).
\]
Moreover, 
\[
kR\ge H(M|K)\ge I(X^{n};M|K)=I(X^{n};M,K)=I(X^{n};W).
\]
 
\end{IEEEproof}

\subsection{Evaluation of the Multi-letter Bound}

Recall that 
\[
U=\Cov(\E[X^{n}|W])=\Cov(a_{W}),
\]
and 
\[
V=\Cov(\E[Y^{n}|W])=\Cov(b_{W}).
\]
Let $u_{i},1\le i\le n$ be eigenvalues of $U$. From \eqref{eq:-3},
we know $\rho^{2}\leq u_{i}\leq1.$

From \eqref{eq:-1} and \eqref{eq:-2}, we get 
\begin{align}
\E H_{W} & \leq n\log(2\pi e)+\frac{1}{2}\log\left[\det(I_{n}-U)\det(I_{n}-V)\right]\nonumber \\
 & =n\log(2\pi e)+\frac{1}{2}\sum^{n}_{i=1}\log\left[(1-u_{i})\left(1-\frac{\rho^{2}}{u_{i}}\right)\right].\label{eq:-4}
\end{align}

Moreover, we have 
\begin{align*}
I(X^{n};W) & =h(X^{n})-h(X^{n}|W)\\
 & \ge\frac{n}{2}\log(2\pi e)-\mathbb{E}\bigl[\frac{1}{2}\log\{(2\pi e)^{n}\det(\Cov(X^{n}\bigm|W))\}\bigr]\\
 & =-\frac{1}{2}\mathbb{E}\bigl[\log\det(A_{W})\bigr]\\
 & \ge-\frac{1}{2}\log\det(\mathbb{E}[A_{W}])\\
 & =-\frac{1}{2}\log\det(I_{n}-U)\\
 & =-\frac{1}{2}\sum^{n}_{i=1}\log(1-u_{i}),
\end{align*}
where the first inequality follows since Gaussian distributions maximize
entropy, the second inequality follows by the concavity of $\log\det$,
and the third equality follows by the law of total covariance. 

We next show that setting $u_{i},1\le i\le n$ to the same value
will lead to a further bound. 
\begin{lem}
\label{lem:same}Assume that $\rho^{2}\le u_{i}\le1$ for all $i$.
If 
\[
\frac{1}{n}\sum^{n}_{i=1}\log(1-u_{i})=\log(1-\alpha),
\]
then 
\[
\frac{1}{n}\sum^{n}_{i=1}\log\left(1-\frac{\rho^{2}}{u_{i}}\right)\le\log\left(1-\frac{\rho^{2}}{\alpha}\right).
\]
\end{lem}
\begin{IEEEproof}
Define $x_{i}=\log(1-u_{i}),$ so that the assumption is equivalent
to 
\[
\frac{1}{n}\sum^{n}_{i=1}x_{i}=\log(1-\alpha).
\]

Consider the function 
\[
f(x)=\log\left(1-\frac{\rho^{2}}{1-e^{x}}\right),
\]
defined for $x\le\log(1-\rho^{2})$. Since $u=1-e^{x}$, we have 
\[
f(\log(1-u))=\log\left(1-\frac{\rho^{2}}{u}\right).
\]

We first show that $h$ is concave. A direct calculation gives 
\[
f''(x)=e^{x}\left(\frac{1}{(1-e^{x})^{2}}-\frac{1-\rho^{2}}{(1-\rho^{2}-e^{x})^{2}}\right).
\]
Because $u=1-e^{x}\ge\rho^{2}$, we have $e^{x}\le1-\rho^{2}.$ Moreover,
\[
(1-\rho^{2}-e^{x})^{2}\le(1-\rho^{2})(1-e^{x})^{2},
\]
which implies 
\[
\frac{1}{(1-e^{x})^{2}}\le\frac{1-\rho^{2}}{(1-\rho^{2}-e^{x})^{2}}.
\]
Hence $f''(x)\le0,$ and therefore $f$ is concave.

By Jensen's inequality for concave functions, 
\[
\frac{1}{n}\sum^{n}_{i=1}f(x_{i})\le f\left(\frac{1}{n}\sum^{n}_{i=1}x_{i}\right).
\]
Using the definition of $x_{i}$ and the constraint, we obtain 
\[
\frac{1}{n}\sum^{n}_{i=1}\log\left(1-\frac{\rho^{2}}{u_{i}}\right)\le\log\left(1-\frac{\rho^{2}}{\alpha}\right).
\]
\end{IEEEproof}
Applying Lemma \ref{lem:same} yields that for some $\alpha$ such
that $\rho^{2}\le\alpha\le1$, it holds that 
\begin{align*}
I(X^{n};W) & \ge-\frac{1}{2}\log(1-\alpha),
\end{align*}
and 
\begin{align}
\frac{\E H_{W}}{n} & \leq\log(2\pi e)+\frac{1}{2}\log\left[(1-\alpha)\left(1-\frac{\rho^{2}}{\alpha}\right)\right].\label{eq:-5}
\end{align}

Analogously to \eqref{eq:Gamma-t}, we have 
\begin{align*}
R+R_{0}\geq & \log(2\pi e\sqrt{1-\rho^{2}})+\left(-1+\frac{\rho t}{1-\rho^{2}}\right)\frac{\E H_{W}}{n}-\frac{\rho t}{1-\rho^{2}}\log(2\pi t).
\end{align*}
Substituting \eqref{eq:-5} into this inequality yields
\begin{align}
R+R_{0}\geq{} & \log(2\pi e\sqrt{1-\rho^{2}})-\frac{\rho t}{1-\rho^{2}}\log(2\pi t)\nonumber \\
 & +\left(-1+\frac{\rho t}{1-\rho^{2}}\right)\left(\log(2\pi e)+\frac{1}{2}\log\left[(1-\alpha)\left(1-\frac{\rho^{2}}{\alpha}\right)\right]\right).\label{eq:Gamma-opt}
\end{align}

The optimal parameter $t$ is 
\[
t=\sqrt{(1-\alpha)\left(1-\frac{\rho^{2}}{\alpha}\right)}.
\]
Substituting it into \eqref{eq:Gamma-opt} yields the desired bound
\[
R+R_{0}\ge\frac{1}{2}\log\frac{1-\rho^{2}}{(1-\alpha)\left(1-\frac{\rho^{2}}{\alpha}\right)}+\frac{\rho\sqrt{(1-\alpha)\left(1-\frac{\rho^{2}}{\alpha}\right)}}{1-\rho^{2}}.
\]

\subsection*{Acknowledgements}

\emph{Generative-AI use disclosure:} During the preparation of this
manuscript, the author used ChatGPT (GPT-5.5-mini model, Think mode)
as an auxiliary tool for exploring proof ideas, checking calculations,
and improving exposition. All mathematical content, including statements,
proofs, and references, was independently verified by the author,
who assumes full responsibility for the final manuscript.

\bibliographystyle{unsrt}
\bibliography{ref}

\end{document}